\documentclass[12pt,a4paper]{amsart}

\usepackage[english]{babel}
\usepackage[T1]{fontenc}

\usepackage{
    amsthm,
    amsfonts,
    amssymb,
    ifthen,
    amsmath,
    graphicx,
    quiver,
indentfirst,
    mathrsfs,
    accents,
    amsaddr,
    lilyglyphs,
    multirow,
    mathrsfs,
    textcomp,
    manyfoot,
    mathtools,
    soul}
\usepackage[utf8]{inputenc}
\usepackage[svgnames]{xcolor} 
\usepackage[colorlinks,citecolor=DarkGreen,urlcolor=DarkViolet, bookmarksdepth=0, linkcolor=Purple]{hyperref} 

\usepackage[title]{appendix}
\usepackage[
    protrusion=true,
    activate={true,nocompatibility},
    final,
    tracking=true,
    kerning=true,
    spacing=true,
    factor=1100]{microtype}
\SetTracking{encoding={*}, shape=sc}{40}

\usepackage{geometry}
\usepackage[backend=biber,style=alphabetic]{biblatex}
\newtheorem{thm}{Theorem}[section]   
\newtheorem{defn}[thm]{Definition}    
\newtheorem{example}[thm]{Example}   

\makeatletter
\def\widebreve{\mathpalette\wide@breve}
\def\wide@breve#1#2{\sbox\z@{$#1#2$}%
     \mathop{\vbox{\m@th\ialign{##\crcr
\kern0.08em\brevefill#1{0.8\wd\z@}\crcr\noalign{\nointerlineskip}%
                    $\hss#1#2\hss$\crcr}}}\limits}
\def\brevefill#1#2{$\m@th\sbox\tw@{$#1($}%
  \hss\resizebox{#2}{\wd\tw@}{\rotatebox[origin=c]{90}{\upshape(}}\hss$}
\makeatletter

\def\f{f^{*}}

\def\fl{f^*\mathbb{L}^N}

\def\pq{pseudo-prequantisation\,\;}

\title{Non-Symplectic Deformations of Geometric Quantisation
}

\author{Kerr Maxwell}
\address[A1]{School of Physics \& Astronomy, University of Birmingham, UK, B15 2TT}
\address[A1]{EPSRC CDT in Topological Design, University of Birmingham, UK, B15 2TT}
\email{kxm149@bham.ac.uk}
\date{\today}

\begin{document}

\begin{abstract}
We introduce the notion of geometric pseudo-quantisation based on geometric quantisation with a weakened curvature condition. We show how such a structure arises naturally from simple deformations of the symplectic structure and pullbacks of prequantum data by non-symplectic diffeomorphisms. Our main result is deriving the equations of motion for some simple pseudo-quantisations. We also compute the pseudo-quantisation of a number of simple examples from symplectic and almost-symplectic geometry and the general form of the resulting deformed canonical commutator.
\end{abstract}

\maketitle

\section{Introduction}

Given a symplectic manifold $(M,\omega)$, geometric quantisation (GQ) \cite{woodhouse1992geometric} constructs a Hilbert space $\mathcal{H}$ and map $\,\;\widehat{}\;:C^\infty(M)\rightarrow\mathcal{O}(\mathcal{H})$ satisfying the axioms of Quantum Theory \cite{dirac1981principles,hall2013quantum}. $\mathcal{H}$ is broadly identified with a distinguished set of invariant sections (or higher cohomological objects) of a complex line bundle over $M$; compatible with a special curvature condition. The Hilbert space may be constructed by equipping $M$ with a variety of auxiliary geometric structures --- prequantum data --- defining different \textit{``flavours''} of GQ. Much work has been devoted to understanding how isomorphism classes of $\mathcal{H}$ depend on the choice of prequantum structure and their internal multiplicities, as well as studying their interplay with various maps (e.g., symplectic reduction \cite{hall2007unitarity} and Lagrangian correspondences \cite{bates1997lectures}). 

Prequantum data is non-generic in the sense that they are drawn from classes of structure where only particular elements or subspaces satisfy the conditions for GQ. The aim of this work is to explore the space of quantum theories resulting from a weakening of the constraints on certain prequantum data. In particular, we consider cases where the curvature of the connection on the prequantum line bundle (encoding the \textit{classical analogy} between Poisson brackets and commutators) is modified in one of two natural ways: the connection form is augmented by an additional smooth function $f$ or, the connection is replaced by the pullback connection generated by a (not-necessarily symplectic) diffeomorphism. We call such a framework a geometric \textit{pseudo-quantisation}. Our main result characterises the equations of motion (those generated by the flow of the quantised quadratic momentum observable) for different classes of pseudo-quantisation.
\begin{quotation}
\begin{thm}[Main Result]
    Let $(M,\omega)$ be pseudo-quantised with connection form $\Theta = (1+f)\theta$ and polarisation $P$ where $f$ is a monomial of order $n$. Let $\psi\in\breve{\mathcal{H}}_P$ be an element of the pseudo-quantum Hilbert space. Then, if $f$ is polarised, the flow of the quadratic momentum observable is undefined. If $f$ is polarised with respect to $P^\perp$, the flow generates the following equation of motion:
    \begin{equation}
    i\hbar\frac{d\psi}{dt} = -\frac{\hbar^2}{2}\frac{1}{(1+2q^n)^{3/2}}\frac{d^2 \psi}{dq^2}.
\end{equation}
\end{thm}
\end{quotation}
These results are achieved by computing the BKS pairing map (Theorem \ref{thm:main1} for the polarised case and Theorem \ref{thm:main2} for the perpendicular polarised case) whose local formulation extends without issue from standard GQ to pseudo-quantisation.

The motivation for this work is twofold. Mathematically, various loosenings of symplectic structures have received interest in recent years, particularly $b^m$-symplectic manifolds \cite{guillemin2014symplectic,kiesenhofer2017cotangent} and folded/origami symplectic manifolds \cite{da2000unfolding,da2011symplectic}, with efforts underway to extend purely symplectic tools and results to these new settings \cite{kiesenhofer2016action,gualtieri2017tropical,guillemin2017convexity, miranda2021singular, brugues2024arnold}. These manifolds can naturally result from non-canonical coordinate transformations of mechanical systems \cite{braddell2019invitation} or surjective smooth maps (foldings). We are interested in exploring the \textit{quantisation} of such structures, complementing the existing results \cite{guillemin2021geometric,mir2023bohr} with a more physical perspective. Physically, we are motivated by the ease with which the modification of the curvature condition in geometric quantisation leads to non-canonical commutators. Such relations, like 
\begin{equation}
    [\widehat{\mathbf{p}},\widehat{\mathbf{q}}]= -i\hbar(1+\alpha\mathbf{q}^2+\beta\mathbf{p}^2)\,\,\text{for}\,\,\alpha,\beta\in\mathbb{R},
\end{equation}
studied in Ref \cite{kempf1995hilbert} give rise to so-called generalised uncertainty principles \cite{tawfik2015review}, important for the understanding of theories with extra fundamental length scales (such as some quantum gravity scenarios \cite{hossenfelder2013minimal}) or singular behaviour (such as the \textit{Dirac cone} in condensed matter physics \cite{hasan2010colloquium}). In this work, we demonstrate how these generalised structures may be incorporated into a familiar quantum setting in a natural geometric fashion, without affecting the underlying classical theory.

The work proceeds as follows. In \hyperref[sec:ppq]{Section \ref{sec:ppq}} we introduce pseudo-prequantisation and discuss its relation to pullback structures. In \hyperref[sec:pq]{Section \ref{sec:pq}} we incorporate polarisations to define pseudo-quantisation and study directly pseudo-quantisable observables. In \hyperref[sec:bks]{Section \ref{sec:bks}} we calculate the BKS pairing generated by quadratic momentum observables and derive the corresponding Schr\"odinger equations, if they exist. Throughout, we make an effort to discuss concrete examples and provide versions of important technical expressions in local coordinates.

\section{Pseudo-prequantisation}\label{sec:ppq}
\subsection{Standard Prequantisation}
We briefly recall the basic construction of prequantisation. A complete presentation of the theory can be found in a number of reference works \cite{sniatycki2012geometric,woodhouse1992geometric,hall2013quantum}. In what follows, $(M,\omega)$ is a $2n$-dimensional symplectic manifold. 

\begin{defn}\label{defn:pgq}
A prequantisation of $(M,\omega)$ consists of the tuple
\begin{equation}
    (\mathbb{L},h,\nabla)_M
\end{equation}
Where $\mathbb{L}$ is complex line bundle over $M$, $h$ is a hermitian inner product on the fibres of $\mathbb{L}$ and $\nabla$ is a hermitian connection on $\mathbb{L}$ such that $\omega=\text{curv}(\nabla)$. 
\end{defn}

We take $\epsilon_\omega = \frac{1}{n!}\omega^n$ to be the conventionally scaled canonical volume form on $M$. The hermitian condition is equivalent to $\omega$ defining an integral cohomology class in $H^2(M,\mathbb{Z})$, though in this work we will mostly be working over $T^*\mathbb{R}^n$ and are more concerned with local behaviour.
\begin{defn}
    Let $(\mathbb{L},h,\nabla)_M$ be a prequantisation. The prequantum Hilbert space $\mathcal{H}_{\mathrm{pre}}$ associated to the prequantisation is $L^2(\Gamma(\mathbb{L}))$, with $\Gamma(\mathbb{L})$ the space of smooth sections, with respect to the inner product
    \begin{equation}
        \langle \psi_1,\psi_2\rangle =\int_Mh(\psi_1,\psi_2)\,\epsilon_\omega\,\,\, \mathrm{for} \,\,\,\psi_1,\psi_2\in\Gamma(\mathbb{L}).
    \end{equation}

\end{defn}

Our connections will always be written with an explicit $\hbar$, i.e. $\nabla = d-\frac{i}{\hbar}\theta$ with respect to a connection $1$-form $\theta$. The curvature form $\text{curv}(\nabla)$ is defined as
\begin{equation}\label{eqn:curvatureDef}
    \text{curv}(\nabla)(X,Y) = i\hbar\left([\nabla_X,\nabla_Y] - \nabla_{[X,Y]}\right)\,\,\,\mathrm{for}\,\,\,X,Y\in \Gamma(TM).
\end{equation}

The purpose of prequantisation is to take the standard Lie algebra homomorphism from smooth function to Hamiltonian vector fields, $X_{\{A,B\}} = [X_A,X_B]$ for $A$, $B\in C^\infty(M)$, and construct a similar map from $C^\infty(M)$ to operators on $\mathcal{H}_{\mathrm{pre}}$, called the prequantum operator prescription. This map, in GQ given by
\begin{equation}
    \widehat{}\;:f\mapsto\widehat{f} \,\coloneqq-i\hbar\nabla_{X_f}+f,
\end{equation}
formally specifies the quantisation $f\mapsto\widehat{f}$ of observables. Not yet present is the requirement that observables preserve the basis of representation of the wavefunction, achieved by the introduction of polarisations in \hyperref[sec:pq]{Section \ref{sec:pq}}. The motivation behind the prequantum curvature condition $\text{curv}(\nabla)=\omega$ is readily seen by computing the commutator of two such operators
\begin{equation}\label{eqn:comm}
\left[\widehat{A},\widehat{B}\right] = -i\hbar\left(-i\hbar\nabla_{X_{\{A,B\}}} -\text{curv}(\nabla)(X_A,X_B) + 2\{A,B\}\right).
\end{equation}
By choosing $\text{curv}(\nabla) = \omega$, \eqref{eqn:comm} simplifies to a statement of Dirac's \textit{classical analogy}, the canonical relationship;
\begin{equation}
    [\widehat{A},\widehat{B}] = -i\hbar\widehat{\{A,B\}},
\end{equation} between Poisson brackets and quantum commutators. Our point of departure from the standard theory is this identification of the curvature of $\nabla$ with the underlying symplectic form $\omega$ on $M$.

\subsection{Pseudo-Prequantisation}
\begin{defn}
    Let $(M,\omega)$ be a $2n$-dimensional symplectic manifold. A pseudo-prequantisation of $M$ consists of the tuple
    \begin{equation}
        (\mathbb{L},h,\nabla, \Omega)_M,
    \end{equation}
    where $\mathbb{L}$ is complex line bundle over $M$, $h$ is a hermitian inner product on the fibers of $\mathbb{L}$, $\nabla$ is a hermitian connection on $\mathbb{L}$ and $\Omega$ is a closed 2-form on $M$ satisfying $\;\Omega=\text{curv}(\nabla)$.
\end{defn}

Note that here it is now $\Omega$, not $\omega$, which must satisfy the integrality condition for the hermiticity of $\nabla$. Depending on the context, we will either define a pseudo-prequantisation using either a curvature form $\Omega$ or a potential $1$-form $d\Theta=\Omega$. Pseudo-prequantisations exhibit much of the same structure as prequantisations, with the notable exception of the canonical commutator, which is now defined using the arbitrary closed 2-form $\Omega$ rather than the symplectic form $\omega$ that comes with $M$.

\begin{defn}
    Let $(\mathbb{L},h,\nabla, \Omega)_M$ be a \pq of $M$. The \pq of a function $A\in C^\infty(M)$ is denoted $\breve{A}$ and define as
    \begin{equation}
    \begin{split}
        \breve{A}&\coloneq-i\hbar\nabla_{X_A}+A\\
        &=-i\hbar\left(X_A-\frac{i}{\hbar}\Theta(X_A)\right)+A,
        \end{split}
    \end{equation}
    where $d\Theta = \Omega$ is a potential for the curvature form.
\end{defn}
It follows that the commutator of the \pq of $A,B\in C^\infty(M)$ is 
\begin{equation}\label{eqn:defComm}
    [\breve{A},\breve{B}] =-i\hbar\left[-i\hbar\nabla_{X_{\{A,B\}}}  -\Omega(X_A,X_B)+2\{A,B\}\right].
\end{equation}
Note that this expression does not immediately produce an analogue of the classical analogy, as there is no requirement for $\Omega$ to define a Poisson bracket.

Let us consider some basic examples of pseudo-prequantisations, defined either via the curvature or connection forms. Consider the case where $[\Omega]$ coincides with $[\omega]$ in $H^2(M,\mathbb{Z})$. Then $\Omega = \omega + d\gamma$ for $\gamma\in\Omega^1(M)$.  If $d\gamma=0$ then this is equivalent to the augmenting of $\theta$ by the addition of a \textit{flat connection} $\Theta = \theta-\gamma$. The overall flat component of the connection $\nabla$ is removable with a gauge transform and does not affect the resulting quantisation. It is to our advantage then to only consider non-flat modifications $\omega\rightarrow\Omega$ or $\theta\rightarrow\Theta$.

\begin{example}
Beginning from the standard prequantisation of $T^*\mathbb{R}^{2}$, consider the \pq defined by the connection $\Theta = \theta + q_1dq_2$. This gives
\begin{equation}
    \Omega = dp_1\wedge dq_1 + dp_2\wedge dq_2 + dq_1\wedge dq_2.
\end{equation}
The resulting nontrivial commutators are
\begin{equation}
    [\breve{p}_i,\breve{q}_i] = -i\hbar\delta_{i,j}\qquad \text{and} \qquad [\breve{q}_1,\breve{q}_2]=i\hbar.
\end{equation}
As expected, the presence of a basis element of $\Omega^2(T^*\mathbb{R}^2)$ in $\Omega$ couples the corresponding pseudo-prequantum operators.
\end{example}

We can go further and consider arbitrary couplings of the observables 
\begin{example}\label{ex:simple}
    Take $T^*\mathbb{R}^n$ with the standard symplectic structure and let $\left\{f_i,g_i\right\}_{i=1}^n$ be smooth functions on the $i^{\mathrm{th}}$ pair of canonical coordinates. If
    \begin{equation}
    \Theta = \sum_{i=1}^n\left[\left(\frac{p_i}{2}-f_i(p_i,q_i)\right)dq_i-\left(\frac{q_i}{2}-g_i(p_i,q_i)\right)dp_i\right],
    \end{equation}
     then $d\Theta=\Omega = \sum_{i=1}^n\{f_i,g_i\}dp_i\wedge d q_i$ defines a pseudo-prequantisation. The formal canonical commutator in this case is
        \begin{equation}
    [\breve{p}_i,\breve{q}_i]  = 2-\{f_i,g_i\}.
    \end{equation}
\end{example}
Example \ref{ex:simple} demonstrates one of the basic consequences of changing the curvature condition in geometric quantisation, the canonical commutator is generally no longer constant and may even vanish at certain points.

Recent work on $b$-symplectic \cite{guillemin2014symplectic}, folded symplectic \cite{da2000unfolding}, and origami manifolds \cite{da2011symplectic} provides motivating examples of $\Omega$s which have a geometric connection to some underlying genuine symplectic structure. Models of quantisation have also been extended to these structures \cite{guillemin2021geometric,mir2023bohr}. Folded symplectic structures offer a setting for \pq where the pathology is localised topologically and in terms of dimension. We recall the definition and some basic results from \cite{da2000unfolding}.
\begin{defn}
    Let $M$ be a $2n$-dimensional smooth manifold, $\omega\in \Omega^2(M)$ a closed $2$-form and $\iota:Z\hookrightarrow M$ the set along which $\omega^n$ vanishes. If $\omega^n$ intersects the zero section of $\bigwedge^{2n}T^*M$ transversally and $\iota^*\omega$ is of maximal rank then $(\omega,Z)_M$ is a folded symplectic manifold 
\end{defn}
Such a $Z$ is always a codimension $1$ embedded submanifold and possesses Darboux coordinates in the neighbourhood of $Z$ in which $\omega$ may be written
\begin{equation}\label{eqn:foldedrn}
    \omega = p_1 dp_1\wedge dq_1 + \sum_{i=2}^n dp_i\wedge dq_i,
\end{equation}
where $p_1=0$ defines $Z$, also called the \textit{folding hypersurface}. 

\begin{example}[Quantisation in a folded neighbourhood]\label{ex:foldedflat}
Consider $T^*\mathbb{R}^{n}$ with the standard symplectic and prequantum structures. The folded symplectic structure \eqref{eqn:foldedrn} gives rise to the formal relations
\begin{equation}
    [\breve{q}_1,\breve{p}_1] = -i\hbar(2-p_1)\,\,\,\,\text{and}\,\,\,\, [\breve{p}_i,\breve{q}_j] = -i\hbar\delta^i_j \,\,\,\, \text{for}\,\,\,\,i>1.
\end{equation}
For all coordinates except $p_1$ and $q_1$ we have the canonical commutation relations. The odd fact here is the introduction of a new length scale, that separating $Z$ from the hypersurface on which $[\breve{q}_1,\breve{p}_1]=0$.  
\end{example}

\subsection{Pseudo-Prequantisation via Pullback}

The motivation for introducing pseudo-prequantum structures lies in that they are natural and easy to construct, especially if one already has access to a prequantisation. 
\begin{thm}
    Let $f:(M,\omega^M)\rightarrow (N,\omega^N)$ be a diffeomorphism between symplectic manifolds. If $N$ has a prequantum structure $(\mathbb{L},h,\nabla)_N$ then $M$ is naturally pseudo-prequantisable.
\end{thm}
\begin{proof}
The pullback provides a natural choice for each geometric component of the definition. $M$ supports a natural complex line bundle, the pullback bundle $f^*\mathbb{L}$. The hermitian metric $h$ is pulled back pointwise in each fibre. A natural connection on $f^*\mathbb{L}$ is the pullback connection $(f^*\nabla)$. Finally, the pullback $f^*\omega^N$ of the symplectic form on $N$ is a natural choice of closed 2-form. Thus $(M,\omega_M)$ is naturally pseudo-prequantisable with respect to the structure
\begin{equation}
\left(f^*\mathbb{L},f^*h,f^*\nabla,f^*\omega^N\right)_M.
\end{equation}
\end{proof}
If all we have access to is a smooth map from $M$ to $N$, then it is still possible to study local normal forms of pseudo-prequantisations.

\subsection{Pullback of Prequantisation}

We define the pullback of a prequantum operator as the pseudo-prequantisation of the pullback of the underlying observable.

\begin{defn}\label{eqn:pbOperator}
Let $f:(M,\omega^M)\rightarrow (N,\omega^N)$ be a smooth map between symplectic manifolds. Let $N$ be prequantisable and equip $M$ with the natural (possibly local) pseudo-prequantisation induced by $f$. The pullback of a prequantum operator is defined as
\begin{equation}
    \left(f^*\widehat{A}\right) \coloneqq \widebreve{f^*A} = -i\hbar\left(f^*\nabla^N\right)_{X^M_{f^*A}}+f^* A,
\end{equation}
where the superscript on $X$ helps bookkeep which symplectic structure is generating the Hamiltonian vector field. When there is no danger of confusion, we will omit the $f$ and denote the pullback quantisation of $A\in C^\infty(N)$ simply by $\breve{A}$.
\end{defn}

To see explicitly where the mismatch in structure occurs, consider the action of a pullback operator on a section of $f^*\psi\in\Gamma(\fl)$. From the definition of the pullback connection we have
\begin{equation}\label{eqn:pullbackOperator}
    \breve{A} \left(f^*\psi\right) = f^*\left(-i\hbar\nabla^N_{f_*\left(X^M_{f^*A}\right)}\psi + A\psi\right).
\end{equation}
The mismatch occurs precisely because, if $f$ is a local diffeomorphism which is not symplectic, then $f_*(X^M_{\f A})$ and $X^N_A$ only coincide for very specific $f$, e.g. the identity map. The formal commutators occurring in pullback quantisations are explicated by the following theorem.

\begin{thm}\label{eqn:commEqn}
        Let $f:(M,\omega^M)\rightarrow (N,\omega^N)$ be a diffeomorphism between symplectic manifolds of dimension $2n$ and let $M$ be naturally pseudo-prequantised by a prequantisation of $N$. Let $A$, $B\in C^\infty(N)$. The commutator of $\breve{A}$ and $\breve{B}$ is given by
        \begin{equation}\label{eqn:commEqnExp}
        \frac{i}{\hbar}\left[\breve{A},\breve{B}\right] = -i\hbar X^M_{\mathfrak{p}}-(f^*\theta^N)(X^M_{\mathfrak{p}})+\mathfrak{p}\left(2-\sum_{i=1}^n\mathfrak{c}_i\right),
        \end{equation}
        where $\mathfrak{c} _i= \{f^*p^\prime_i,f^*q_i^\prime\}^M$ are the Poisson commutators of the pullback canonical coordinates, $\mathfrak{p}=\{f^*A,f^*B\}^M$ is that of observables and $\theta^N$ is a potential for $\omega^N$.
    \end{thm}

    \begin{proof}
        Substitute the natural pullback structures into \eqref{eqn:defComm} and proceed via direct calculation.
    \end{proof}

The simplest quantisations of nontrivial phase spaces are those of the cylinder (phase space of a particle constrained to a circle) or $S^2$ (whose geometric quantisation gives rise to unitary irreducible representations of SU$(2)$, see \cite{maxwell2025structured} for a detailed discussion). The embedding of a compact cylinder inside a sphere provides a simple setting for a pullback quantisation.

    \begin{example}[Pseudo-Prequantum Squeezed Cylinder]\label{example:pqsc}
    Let $M=S^1\times [-L,L]$ be a cylinder of height $2L$ and $N = S^2_R$ a sphere of radius $R$, both equipped with their standard symplectic structure. Let $f_\lambda:M\rightarrow N$ be a smooth embedding with image the cylindrical strip $[-\frac{L}{\lambda},\frac{L}{\lambda}]\subset S^2_R$, as in Figure \ref{fig:squeezeCylinder}. Such a map only exists when $0 < \lambda < \frac{L}{R}$. If $L< R$ then $f_1$ exists and is symplectic. Varying $\lambda$ about $1$, we are able to \textit{``squeeze"} or \textit{``stretch"} the embedding and resulting quantisation.
    
        \begin{figure}[h]
            \centering
        \includegraphics[width=0.8\linewidth]{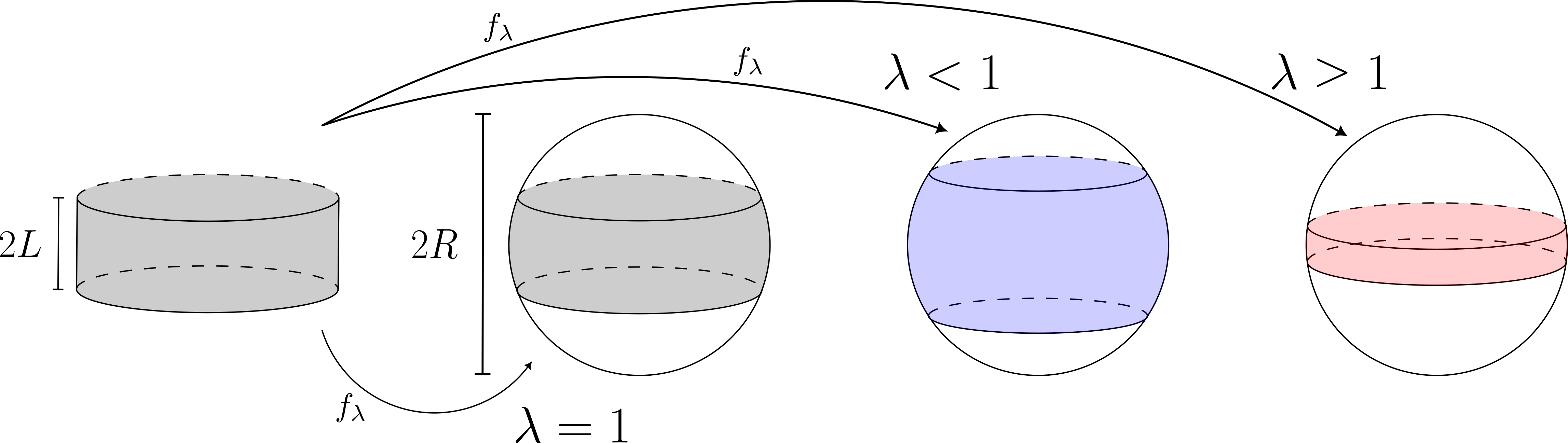}
            \caption{}
            \label{fig:squeezeCylinder}
        \end{figure}

    Let $(l,\phi_l)$ be canonical co-ordinates on $M$ and $(z,\phi_z)$ similarly on $N$. Over $M$ we have the trivial prequantisation, over $N$ we assume $R$ is such that $\omega^N $ satisfies the integrality criteria for the existence of a prequantum line bundle. Decorating commutators by their resident line bundles, we have
    \begin{equation}
        [\hat{l},\hat{\phi}_l]_{\mathbb{L}^M} = -i\hbar\,,\qquad [\hat{z},\hat{\phi}_z]_{\mathbb{L}^N} = -i\hbar.
    \end{equation}
    From theorem \ref{eqn:commEqn} we have
    \begin{equation}
    \left[\breve{z},\breve{\phi}_z\right]_{f^*(\mathbb{L}^N)} = -i\hbar\left(\frac{2\lambda-1}{\lambda^2}\right).
    \end{equation}
    This commutator vanishes as $\lambda\rightarrow\infty$ (though we refrain from interpreting this as some kind of semiclassical limit), vanishes at $\lambda=1/2$ and reverses sign at $\lambda=\sqrt{2}-1$.
    \end{example}

\subsection{Local Coordinates}

To calculate with these structures, it is helpful to have a more explicit local forms from which we can form differential equations. Some new notation is also necessary at this point: in our now standard setup $f:M\rightarrow N$ with coordinates $(p,q)$ on $M$ and $(p^\prime,q^\prime)$ on $N$, denote the pullbacks of functions $A\in C^\infty(N)$ by $\tilde{A}$, this includes the coordinates $\tilde{p}\coloneqq f^*p^\prime$ and $\tilde{q}\coloneqq f^*q^\prime$.

The connection form then decomposes as
\begin{equation}\label{eqn:pbConnectionLocal}
\begin{split}
(f^*\theta^N)&=\tilde{\theta}^N_{p^\prime}d\tilde{p} + \tilde{\theta}^N_{q^\prime}d\tilde{q}\\
&=\left(\tilde{\theta}^N_{p^\prime}\frac{\partial \tilde{p}}{\partial p} + \tilde{\theta}^N_{q^\prime}\frac{\partial \tilde{q}}{\partial p} \right)dp + \left(\tilde{\theta}^N_{p^\prime}\frac{\partial \tilde{p}}{\partial q} + \tilde{\theta}^N_{q^\prime}\frac{\partial \tilde{q}}{\partial q} \right)dq,
\end{split}
\end{equation}
where $\theta^N_{q^\prime}$ and $\theta^N_{p^\prime}$ are the coefficient functions of the 1-form $\theta^N$.
After expanding the Hamiltonian vector field $X^M_{\tilde{A}}$ generated by an observable $\tilde{A}$, the contraction takes the simple form
\begin{equation}\label{eqn:pbContractionLocal}
    \tilde{\theta}^N(X_{\tilde{A}}) = \tilde{\theta}^N_{p^\prime}\{\tilde{A},\tilde{p}\}^M + \tilde{\theta}^N_{q^\prime}\{\tilde{A},\tilde{q}\}^M.
\end{equation}
From this, we can write down a fully local expression for \eqref{eqn:pullbackOperator} for the pullback-induced pseudo-prequantisation of $A\in C^\infty(N)$. Let $\psi\in\Gamma(\mathbb{L})$, then
\begin{equation}
\breve{A}\tilde{\psi} =-i\hbar\left(\frac{\partial \tilde{A}}{\partial p}\frac{\partial \tilde{\psi} }{\partial q}-\frac{\partial \tilde{A}}{\partial q}\frac{\partial\tilde{\psi}}{\partial p}\right)-\left(\tilde{\theta}^N_{p^\prime}\{\tilde{A},\tilde{p}\}^M + \tilde{\theta}^N_{q^\prime}\{\tilde{A},\tilde{q}\}^M-\tilde{A}\right)\tilde{\psi}
\end{equation}

Given \eqref{eqn:pbConnectionLocal} and \eqref{eqn:pbContractionLocal}, it is straightforward, though a little messy, to decompose \eqref{eqn:commEqnExp} into a purely local form,
\begin{equation}\label{eqn:commEqnExp}
        \left[\breve{A},\breve{B}\right] = -i\hbar \left(-i\hbar X^M_{\mathfrak{p}}-(f^*\theta^N)(X^M_{\mathfrak{p}})+\mathfrak{p}\left(2-\sum_{i=1}^n\mathfrak{c}_i\right)\right),
        \end{equation}
At this point, the local expressions are becoming a bit too cumbersome; in the next section, we will begin making assumptions about the form of deformations and their behaviour relative to polarisations.

\section{Pseudo-Quantisation}\label{sec:pq}

\subsection{Polarisation}

Pseudo-prequantisation does not affect the construction of $\mathcal{H}_{\mathrm{pre}}$, thus we encounter the same problem in the standard theory that many elements of $\mathcal{H}_{\mathrm{pre}}$ are \textit{prima facie} unphysical. One standard method for addressing this is to form a new Hilbert space comprised only of sections which are covariantly constant along a global specification of ``position'' or ``momentum'' coordinates (or any complex combination thereof). This information is encoded in structures known as polarisations, which are closely related to Lagrangian foliations. In this section, we review the role of polarisation in producing a true quantum Hilbert space from $\mathcal{H}_{\mathrm{pre}}$ and extend this idea to define a new pseudo-quantum Hilbert space. The question of which observables preserve the pseudo-quantum Hilbert space is also addressed.

\begin{defn}\label{defn:pol}
    Let $M$ be a symplectic manifold. A polarisation $P\subset T^\mathbb{C}M$ is a (complexified) integrable Lagrangian distribution.
\end{defn}
Though much progress had been made in incorporating more singular cases \cite{hamilton2010locally,hamilton2010geometric}, in this work we restrict ourselves to reducible polarisations, where the leaves are fibrating and the natural projection is a submersion onto a manifold.
Let $M$ be a symplectic manifold with a Lagrangian foliation $\mathcal{F}$. Then a full-rank distribution $P\subset T\mathcal{F}\subset T^\mathbb{C}M$ is a real polarisation on $M$. The quantum Hilbert space with respect to a polarisation $P$ is thus defined as
\begin{equation}\label{eqn:qhilb}
    \mathcal{H}_P\coloneq\left\{\psi\in\Gamma(\mathbb{L})\,\,|\,\,\psi\in\mathrm{ker}(\nabla_{X})\,\,\text{for all }X\in\overline{P}\,\,\text{and} \,\,||\psi||_{M/P}<\infty\right\},
\end{equation}
where $\nabla_X$ above is understood to be acting on $\mathcal{H}_{\mathrm{pre}}$ and the norm is over the space of leaves $M/\mathcal{F}$. Elements of this kernel are known as \textit{polarised} or \textit{flat} sections. A common problem encountered in GQ is that unless $M/P$ comes equipped with a natural volume form defining an inner product, it is not always clear how to take the sections admitted to $\mathcal{H}_P$ and complete them to a full Hilbert space. Even then, the Hilbert space may still be trivial if the leaves of $P$ are non-compact.

A complete treatment of the construction of $\mathcal{H}_P$ for real polarisations necessitates the introduction of $\frac{1}{2}$-forms, which in some cases fix the problem of ill-defined Hilbert spaces. The remaining examples all we consider over either $S^2$ or $T^*\mathbb{R}^n$, where the $\epsilon_\omega$ ``split'' nicely along the $P$ and $P^\perp$  considered, i.e. we may simply consider $S^2/\partial\phi$ and $T^*\mathbb{R}^n/\partial_{\mathbf{p}}$ as having volume forms $dz$ and $\bigwedge_{i=1}^ndq_i$ respectively. Without loss of generality we may forgo explicit consideration of $\frac{1}{2}$-forms, relying the previously stated volume forms until \hyperref[sec:bks]{Section \ref{sec:bks}}, where the trivial bundle of polarised $\frac{1}{2}$-forms over $T^*\mathbb{R}$ is considered.

Continuing with our theme of simple generalisations, we take the pseudo-quantum Hilbert space to be defined similarly to its standard counterpart, albeit with a different connection.
\begin{defn}[pseudo-quantum Hilbert space]
Let $M$ be a symplectic manifold with pseudo-prequantisation $(\mathbb{L},h,\nabla,\Omega)_M$. The pseudo-quantum Hilbert space $\breve{\mathcal{H}}_P$ is defined as 
    \begin{equation}\label{eqn:pqhilb}
    \breve{\mathcal{H}}_P\coloneq\left\{\psi\in\Gamma(\mathbb{L})\,\,|\,\,\psi\in\mathrm{ker}(\nabla_{X})\,\,\text{for all}\,\,X\in\overline{P} \,\,\text{and} \,\,||\psi||_{M/P}<\infty\right\}.
\end{equation}
\end{defn}
Note that $\mathcal{H}_P$ and $\breve{\mathcal{H}}_P$ both have their Hilbert space structure defined by the inner product determined by the symplectic structure on $M$, that is, they are integrated against $\epsilon_\omega$. $\Omega$ certainly alters the local form of the sections, but the conditions \eqref{eqn:qhilb} and \eqref{eqn:pqhilb}  are otherwise identical.

\begin{example}[Pseudo-Quantisation of Folded $T^*\mathbb{R}^n$] 
Consider the setup from \hyperref[ex:foldedflat]{Example \ref{ex:foldedflat}}. In the trivialisation with connection form $\Theta = \frac{1}{2}p_1^2dq_1 +\sum_{i=2}^np_idq_i$, polarised sections in $\breve{\mathcal{H}}_{\partial_{\mathbf{p}}}$ and $\breve{\mathcal{H}}_{\partial_{\mathbf{q}}}$ take the general forms 
    \begin{align}
        \psi=F(\mathbf{q}),&\qquad\text{(momentum polarisation)}\\
        \psi=F(\mathbf{p})\exp\left\{\frac{i}{\hbar}\left(\frac{1}{2}p_1^2dq_1 +\sum_{i=2}^np_idq_i\right)\right\},&\qquad\text{(position polarisation)}
    \end{align}
    where in both cases $F$ is a square-integrable function and the norm is $L^2$. 
\end{example}

    \begin{example}[Bohr-Sommerfeld Leaves on Folded $S^2$]
    The unit sphere $S^2$ is prequantisable when equipped with a symplectic form $\omega=E\,dl\wedge d\phi$ satisfying $\frac{1}{2\pi\hbar}\int_{S^2}\omega\in\mathbb{Z}$. In the standard canonical coordinates $(z,\phi)$, $\partial_\phi$ describes a (singular) polarisation with compact leaves. If we take $\breve{H}^1(S^2;\mathcal{J}_{\partial_\phi})$ as the first sheaf cohomology computed with respect to the pseudo-quantisation connection, we can compare the lattices of integral points determining $\mathcal{H}_{\partial_\phi}$ and $\breve{\mathcal{H}}_{\partial_\phi}$ (see \cite{hamilton2010locally} for a discussion of the method). In the standard real quantisation\footnote{In this form of cohomological quantisation, the boundary points, corresponding to the elliptic singularities in the level sets of $l$ do not contribute, unlike in the K\"ahler case.}, $\mathrm{dim}(\mathcal{H}_{\partial_\phi})=2E-1$ with every point on a half-integer lattice being integral. For the \pq we have $\mathrm{dim}(\breve{\mathcal{H}}_{\partial_\phi})$ given by the number of solutions to $\frac{1}{2}(E^2-l^2)\in\mathbb{Z}$ where $|l|<E$. The lattices of integral points are shows in Figure \ref{fig:lattice}.
    \end{example}

        \begin{figure}[h]\label{fig:lattice}
        \centering
        \includegraphics[width=\linewidth]{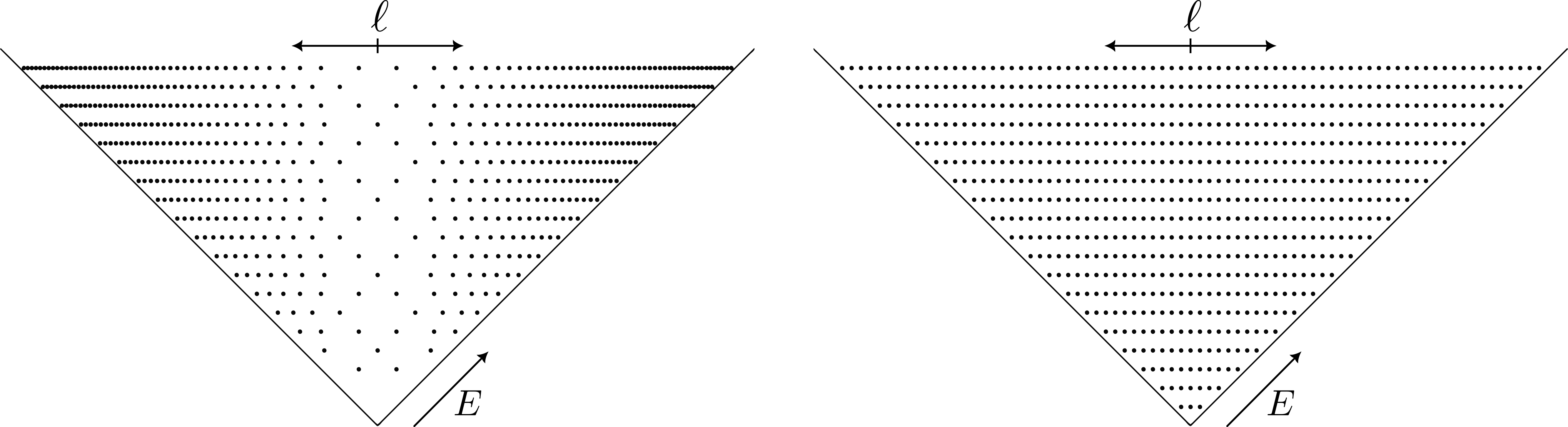}
        \caption{The integral points of the moment map for the \textit{folded} $S^2$ (left) and \textit{standard} $S^2$ (right).}
        \label{fig:placeholder}
    \end{figure}

\subsection{Polarisation Preservation}

Perhaps the most consequential feature of a polarisation is that it determines the set of quantisable observables, that is, observables $\widehat{A}$ are admissible iff $[\nabla_{X},\widehat{A}\,]\psi=0$ for all $\psi\in\mathcal{H}_P$ and $X\in \overline{P}$. In general there are two ways to construct admissible observables; first, one may show that the operator $\widehat{A}$ satisfies the above commutator (such observables are said to be \textit{directly quantisable}), or one may construct a unitary projection of $\widehat{A}\psi$ in $\mathcal{H}_{\mathrm{pre}}$ back into $\mathcal{H}_{P}$. In the standard theory on $T^*\mathbb{R}^n$, it is well known (see e.g. \cite{hall2013quantum}) that only observables at most linear in canonical momentum are directly quantisable, with quadratic momentum observables (including the standard kinetic term) quantisable in the second manner via the BKS pairing. In this section, we compute the admissible directly quantisable observables on $T^*\mathbb{R}^n$ for a general pseudo-quantisation.

If $(M,\omega)$ is a quantisable symplectic manifold with polarisation $P$, then the prequantum lift of a vector field $X$ preserves $P$ if $[X,Y]\in P$ for all $Y\in P$; there is an analogous condition for pseudo-quantisation.
\begin{thm}\label{thm:preserv}
Let $(M^{2n},\omega)$ be pseudo-prequantisable with curvature form $\Omega$ and $\{\beta_i\}_{i=1}^n$ a collection of smooth functions such that $\text{span}(\left\{X_{\beta_i}\right\}_{i=1}^n)$ form a reducible polarisation $P$. An observable $\breve{\alpha}$ preserves $\breve{\mathcal{H}}_P$ iff
\begin{equation}\label{eqn:polMorphism}
\widebreve{\left\{\alpha,\beta_i\right\}}\psi=\Omega(X_\alpha,X_{\beta_i})\;\psi,
\end{equation}
for each $X_{\beta_i}\in \overline{P}$, for every $\psi\in\breve{\mathcal{H}}_P$.
\end{thm}
\begin{proof}
Take $\psi\in \breve{\mathcal{H}}_P$. By definition, $\breve{\alpha}$ preserves the polarisation iff for all $\beta_i$
        \begin{equation}
            [-i\hbar\nabla_{X_\alpha}+\alpha,\nabla_{X_{\beta_i}}]\psi=0.
        \end{equation}
        The result then follows directly from \eqref{eqn:curvatureDef}.
\end{proof}

It is obvious from \eqref{eqn:polMorphism} that $\breve{\mathcal{H}}_P$ is preserved by any constant functions as well as $\breve{\beta}_i$. Theorem \ref{thm:preserv} suggests that the cases where pseudo-quantisation defines a Poisson morphism on phase space or where $\Omega$ has Poisson leaves may possess particularly interrogable behaviour. Furthermore, if $\Omega$ is non-degenerate (and thus defines a symplectic structure on $M$), then we can introduce the Poisson bracket $\left\{\lambda(\alpha),\lambda(\beta_i)\right\}_\Omega\coloneq\Omega(X_\alpha,X_{\beta_i})$ and we have the intriguing case of all $\alpha$ being directly quantisable iff we have that the coincidence ``$\lambda(\alpha)=\breve{\alpha}$\,'' occurs and defines a Poisson map on $M$. We do not pursue these perspectives here, preferring a simple analysis of admissibility based on the polynomial dependence of the observables for this first work.

There are two simplifications of \eqref{eqn:polMorphism} which are useful to write down before going further. In the case where $\omega$ and $\Omega$ are cohomologous, i.e. there exists a $1$-form $\gamma$ such that $d\gamma = \omega-\Omega$, then \eqref{eqn:polMorphism} simplifies to
\begin{equation}
     -i\hbar\nabla_{X_{\{\alpha,\beta_i\}}}\psi+d\gamma(X_\alpha,X_{\beta_i})\psi=0.
\end{equation}
If $\Omega$ is a scaling of $\omega$ by some $f\in C^\infty(M)$ taking the form $\Omega = (f+1)\omega$, then \eqref{eqn:polMorphism} can be written as
\begin{equation}\label{eqn:scaledpreservcond}
    -i\hbar\nabla_{X_{\left\{\alpha,\beta_i\right\}}}\psi = f\left\{\alpha,\beta_i\right\}\psi.
\end{equation}

Now let us consider two cases of pseudo-quantisation by scaling based on the dependence of scaling function $f$. In the following cases, we will always let $\alpha_i,\beta_i$ be canonical momentum-position coordinates in a local neighbourhood $i:U\hookrightarrow M$ for which $\omega=d\alpha_i\wedge d\beta_i$. For $A=\alpha_i^m\beta_i^n$ with $m,n\in\mathbb{N}$, a smooth observable in a single set of canonical coordinates $\alpha_i, \beta_i$, we have
\begin{equation}
    X_{\{A,\beta_j\}} = \frac{\partial^2 a^m}{\partial a^2}\beta^n_i\partial_{\beta_i} -\alpha^m_i\frac{\partial \beta_i^n}{\partial \beta}\partial_{\alpha_i}.
\end{equation}
We always work with respect to an adapted connection (i.e. $\theta(X_{\beta_i})=0$), which admits to $\breve{\mathcal{H}}_P$ only those square-integrable functions depending only on the $\beta_i$. We can now consider some specific cases.

\subsubsection{Standard Case}
Recall that the directly quantised observables in standard geometric quantisation are those observables $A$ satisfying
\begin{equation}\label{eqn:basicpolcond}
    -i\hbar\nabla_{X_{\{A,\beta_i\}}} \psi= 0.
\end{equation}
From this, \eqref{eqn:basicpolcond} may also be expressed as the condition that
\begin{equation}
    \frac{\partial}{\partial \alpha_i}\left\{A,\beta_i\right\}=0, 
\end{equation}
the well known constraint the $A$ may only contain terms at most linear in $\alpha_i$. The general lesson drawn from \eqref{eqn:basicpolcond} is that any $\alpha^{\geq2}$ terms enter coupled to $\psi$ and its derivatives.
\subsubsection{Polarised Scaled Case}
Let us begin with the simplest possible case, where $\Omega = (1+f)\omega$ and $f$ is a polarised function (i.e. $f=f(\underline{\beta})$ so $X_{\beta_i}f=0$). Then we can write $\Theta=(1+f)\theta$ in the trivialisation of $\mathbb{L}$ in which the connection is adapted to $P$, (i.e. $\theta(X_{\beta_i})=\Theta(X_{\beta_i})=0$) being of form $\theta = \alpha_id\beta_i$. The $\psi\in\breve{\mathcal{H}}_P$ are square integrable sections given locally by $\psi = F(\underline{\beta})$ and \eqref{eqn:scaledpreservcond} for an observable $A = \alpha_i^m\beta^n_i$ is
\begin{equation}\label{eqn:modedscalarpreserv}
-i\hbar\frac{\partial^2\alpha_i^m}{\partial \alpha^2_i}\frac{\partial\psi}{\partial\beta_i} = \left(\frac{mf}{m-1}+1\right)\frac{\partial\alpha_i^m}{\partial \alpha}\psi,\quad\mathrm{for\,\,} m\geq 2.
\end{equation}
The cases $m=0,1$ with arbitrary $n$ preserve $\breve{\mathcal{H}}_P$ as the vanishing derivatives make \eqref{eqn:scaledpreservcond} trivial. Under our simplified assumptions on $f$,  \eqref{eqn:modedscalarpreserv} is only satisfied if $\psi$ is trivial. The directly quantisable observables in the pseudo-quantisation with respect to $\Omega = (1+f)\omega$ with a polarised $f$ are thus the same as those under the standard quantisation with respect to $\omega$.

\subsubsection{General Scaled Case}

    Consider this same situation, but now let $f\in C^\infty(M)$. Take the curvature form to be 
    \begin{equation}
        \Omega =d\left([1+f]\theta\right) = (1+f)\omega + df\wedge \theta.
    \end{equation}
    In the same adapted trivialisation as before, we still have polarised sections as being locally given by $L^2$ functions of $\beta$ only. Letting our observable $A$ enter via $\mathfrak{p} = \{A,\beta_i\}$, straightforward calculation shows that the polarisation is preserved by those $A$ giving rise to $\mathfrak{p}$ satisfying
    \begin{equation}
        \sum_{i=0}^n\frac{\partial\mathfrak{p}}{\partial \alpha_i}\left( \left(1+f+\frac{\partial f}{\partial\alpha_i}\right)\alpha_i\,\psi + i\hbar\frac{\partial\psi}{\partial\beta_i}  \right) + f\,\mathfrak{p}\,\psi=0.
    \end{equation}
    There are some immediate contrasts with the case of polarised $f$. If $\mathfrak{p}$ is a constant, i.e. a canonical commutator, then only $f=0$ allows the preservation of the Hilbert space. Any higher-order observables, naturally, depend on $\psi$.

\section{BKS Pairings of Pseudo-quantisations}\label{sec:bks}
To complete our analysis of geometric pseudo-quantisation, in this section, we construct and compute the BKS maps for quadratic momentum observables in pseudo-quantisations generated by the scaling function $f$ being polarised or anti-polarised. The reader is referred to \cite{woodhouse1992geometric} or \cite{sniatycki2012geometric} for background on the BKS construction.\\

The BKS pairing \cite{woodhouse1992geometric} is the statement 
\begin{equation}\label{eqn:bks}
        \left\langle \frac{d\widetilde{\psi}}{dt},\widetilde{\psi}^\prime\right\rangle = \lim_{\tau\rightarrow0}\left[-\frac{d}{dt^{\prime}}\left\langle\!\!\left\langle \left(\widetilde{\rho}^{\,A}_{\tau}\right)^*\widetilde{\psi},\widetilde{\psi}^\prime\right\rangle\!\!\right\rangle\right],
\end{equation}
where the LHS is an inner product on $\breve{\mathcal{H}}_P$ and the RHS a pairing map composed of an inner product over $\breve{\mathcal{H}}_{\mathrm{pre}}$ and the pairing of polarised $\frac{1}{2}$-form sections $\nu,\nu^\prime\in\Gamma(\delta)$ is given by
\begin{equation}
    \langle \sqrt{\nu},\sqrt{\nu^\prime}\rangle = \sqrt{(\nu\otimes\nu)\wedge(\overline{\nu^\prime\otimes\nu^\prime})},
\end{equation}
where $\delta$ is a square root of the canonical bundle (i.e., $\delta\otimes\delta$ is isomorphic to the canonical bundle $\mathcal{K}(M)$) of the symplectic manifold in question. In what follows we have only the trivial $\delta$ over $T^*\mathbb{R}$, for which we take $\sqrt{dq}$ or $\sqrt{d\beta}$ to be the trivialising polarised sections. We remain working in canonical coordinates where either $\omega = dp_i\wedge dq_i$ or $\omega = d\alpha_i\wedge d\beta_i$.

To begin, we review the standard derivation of the Schr\"odinger equation for describing the flow of $\widetilde{\psi} = \psi\otimes\sqrt{dq}$ under a Hamiltonian $H\in C^\infty(T^*\mathbb{R})$ of the form $H=\frac{1}{2}p^2+V(q)$. With respect to the standard prequantum structure, the integral in question is the $dp$ component of the pairing in \eqref{eqn:bks}. After expanding the equation is

\begin{equation}\label{eqn:basicBKS}
    \frac{d\psi}{dt}=-\lim_{\tau\rightarrow 0}\left[\frac{d}{d\tau}\int_{-\infty}^{\infty}\left(\sum_{\lambda=0}^\infty\frac{(\tau p)^\lambda}{\lambda!}\frac{d^\lambda \psi}{d q^\lambda}\right)\sqrt{\tau}\,e^{\frac{i}{\hbar}(\frac{1}{2}p^2-V(q))\tau}dp\right].
\end{equation}
The $\lambda=1$ term vanishes by symmetry of the integral and the $\lambda=0,2$ contributions may be calculated by manipulating \eqref{eqn:basicBKS} into a generalised (even) Fresnel integral or the method of stationary phase (see \cite{sniatycki2012geometric} for more details), with higher order terms always vanishing in the limit. Evaluating the integral gives
\begin{equation}\label{eqn:almostTDSE}
    i\hbar\frac{d\psi}{dt} = -\sqrt{2\pi\hbar}\,e^{i\frac{\pi}{4}}\left(-\frac{\hbar^2}{2}\frac{d^2\psi}{dq^2} + V(q)\,\psi\right).
\end{equation}
The prefactor may be absorbed into the definition of the pairing to yield the standard Schr\"odinger equation for $H$.

We take up now the case of the local theory of a general deformation of the connection of the form $\Theta = (1+f)\theta$ for $f\in C^\infty(T^*\mathbb{R})$. In the polarisation generated by $X_{\beta}$ and in the adapted gauge, polarised sections are given by
\begin{equation}
    \widetilde{\psi} = \psi(\beta)\otimes\sqrt{\beta}
\end{equation}
where $\sqrt{\underline{\beta}}$ is a unit section of the trivial bundle of $X_{\beta_i}$ polarised $\frac{1}{2}$-forms.

\begin{thm}[No Schr\"odinger Equation for Momentum Deformations]\label{thm:main1}
Let $\Theta = (1+f)\theta$ define a pseudo-quantisation of $(T^*\mathbb{R}^n,d\alpha\wedge d\beta)$ in the polarisation generated by $X_\beta$ and $f$ be a linear or greater polarised monomial function. Then, for the observable $\frac{1}{2}\widebreve{\alpha^2}$, the BKS pairing does not converge.
\end{thm}

\begin{proof}
Consider the $\sqrt{d\alpha}$ components of the expansion of the flow under $\widetilde{\rho}^A_{t^\prime}$ for an arbitrary observable $A$.
\begin{equation}\label{eqn:fullbksexpansion}
    \begin{split}
\left(\widebreve{\rho}^{\,A}_{\tau}\right)^*\widetilde{\psi}(\underline{\beta}) &= \left(F(\underline{\beta})+\tau\frac{\partial A}{\partial \alpha_\lambda}\frac{\partial F}{\partial \beta_\lambda} + \frac{\tau^2}{2}\left[\left(\frac{d}{dt}\frac{\partial A}{\partial \alpha_\lambda}\right)\frac{\partial F}{\partial\beta_\lambda}+ \left(\frac{\partial A}{\partial \alpha_\lambda}\right)^2\frac{\partial^2F}{\partial\beta^2_\lambda}\right]+\mathcal{O}(\tau^3)\right)\\
        &\cdots\times\exp\left\{\frac{i}{\hbar}\left(-A\tau+\int_0^{\tau}\sum_{k=1}^n\alpha_k(1+f)\frac{\partial A}{\partial \alpha_k}\circ\gamma(t^\prime)\,dt^\prime \right)\right\}\\
        &\cdots \otimes\sum_{i=1}^n\sqrt{\tau\sum_{j=1}^n\frac{\partial^2A}{\partial \alpha_j\partial\alpha_i}d\alpha_j+\mathcal{O}(\tau^2)},
    \end{split}
\end{equation}
noting that $f$ only enters into the phase term. Taking $A=\frac{1}{2}\alpha_\nu^2$ and evaluating the integral in the exponent via the polarisation of $f$, we have 
\begin{equation}\label{eqn:BKSexpansion}
    \begin{split}
    \left(\widebreve{\rho}^{\,\frac{1}{2}\alpha_\nu^2}_{t^\prime}\right)^*\widetilde{\psi}(\underline{\beta}) &= \left(F(\underline{\beta})+\tau\alpha_\nu\frac{\partial F}{\partial\beta_\nu}+\frac{1}{2}{\tau}^2\alpha^2_\nu\frac{\partial^2F}{\partial\beta^2_\nu}+\mathcal{O}({\tau}^3)\right)\\
        &\cdots\times\exp\left\{\frac{i}{2\hbar}\alpha^2_\nu \tau(1+2f)\right\}\otimes\sqrt{\tau d\alpha_\nu+\mathcal{O}(\tau^2)}.
    \end{split}
\end{equation}
Let the monomial deformation take the form $f=\frac{\lambda}{2}\alpha_\nu^n$. Denote the resulting integral contribution to \eqref{eqn:bks} proportional to $\left(\partial^m_{\beta_\nu}F\right)\,\overline{F^\prime}$ (both independent of $\alpha$) as $I_{n,m}$. These contributions take the form
\begin{equation}
    I_{n,m} = -\lim_{\tau\rightarrow 0}\frac{d}{d\tau}\left[\int_{-\infty}^\infty \frac{\tau^{m+1/2}}{m!}\exp\left\{\frac{i\tau}{2\hbar}\alpha^2_\nu(1+\lambda\alpha_\nu^n)\right\}d\alpha_\nu\right].
\end{equation}
The general pattern for evaluating these expressions is to isolate $\tau$ to the non-dominant term in the phase and expand as a power series. Let $\mu = \alpha_\nu\,\tau^{\frac{1}{n+2}}$, then the $j$-th term of the resulting expansion is
\begin{equation}\label{eqn:inmj}
\begin{split}
    I_{n,m,j}\coloneq -\lim_{\tau\rightarrow 0}\left[\tau^{-\frac{1}{2}+m-\frac{1}{2n}+\frac{jn}{2+n}}\right]&\left(\frac{1}{j!\,m!}\left(\frac{i}{2\hbar}\right)^j\left(\frac{1}{2}+m-\frac{1}{2n}+\frac{jn}{2+n}\right)\right)\\
    &\cdots\times\int_{-\infty}^\infty \mu^{2j}e^{\frac{i\lambda\mu^{n+2}}{2\hbar}}d\mu.
\end{split}
\end{equation}
For a fixed $n$ and $m$, there is a unique rational $j^\prime$ satisfying the condition
\begin{equation}
    j^\prime= \frac{(n+2)(n-2mn+1)}{2n^2},
\end{equation}
corresponding to the case when the limit in \eqref{eqn:inmj} goes to unity. If $j^\prime$ is a positive integer, then this limit occurs for one of the terms in the expansion: $I_{n,m,j^\prime}$, and is the unique (in the sense of being possibly nonzero) contribution to $I_{n,m}$. Additionally, every limit where $j>j^\prime$ will vanish and every limit where $j<j^\prime$ will either vanish or diverge.
The BKS pairing for the $n^\mathrm{th}$ order deformation can only be well-defined if all $I_{n,m,j}$ are finite\footnote{Though even if the limit with respect to $\tau$ is finite, the contribution can diverge in other ways}.

To prove the theorem, consider that the condition for $I_{n,m,j}$ to diverge in the limit $\tau\rightarrow 0$ is
\begin{equation}
    -\frac{1}{2}+m-\frac{1}{2n}+\frac{jn}{2+n}<0.
\end{equation}
As $n>0$ it is clear that $I_{n,0,0}$ always diverges in the limit and thus \eqref{eqn:bks} never converges for the observable and deformations in question.
\end{proof}
We now prove the analogous result for deformations normal to the polarisation.
\begin{thm}[Schr\"odinger Equation for Position Deformations]\label{thm:main2}Let $\Theta = (1+f)\theta$ define a pseudo-quantisation of $(T^*\mathbb{R}^n,d\alpha_i\wedge d\beta_i)$ in the polarisation generated by $X_\beta$ and $f$ be a linear or greater monomial function on the leaves of the polarisation (i.e. $f=f(\alpha)$). Then, for the observable $\frac{1}{2}\widebreve{\alpha^2}$, the BKS pairing is finite and generates the following equation of motion for $\frac{1}{2}\widebreve{\alpha^2}$.
\begin{equation}
    i\hbar\frac{dF}{dt} = -\frac{\hbar^2}{2}\frac{1}{(1+2\beta^n)^{3/2}}\frac{d^2 F}{d\beta^2}. 
\end{equation}
\end{thm}
\begin{proof}
    We begin from \eqref{eqn:fullbksexpansion}, resolving the integral in the phase by Taylor-expansion about $t$ gives the explicit BKS pairing
\begin{equation}\label{eqn:main2initial}
\left\langle \frac{d\widetilde{\psi}}{dt},\widetilde{\psi}^\prime\right\rangle =
\lim_{\tau\rightarrow 0}\frac{d}{d\tau}\left[\int d\underline{\beta}\, \overline{F}^\prime\int d\underline{\alpha}\, \sqrt{t}\left(\sum_{\lambda=0}^\infty \frac{\tau^\lambda}{\lambda!}\underline{\alpha}^\lambda\frac{\partial^\lambda F}{\partial\underline{\beta}^\lambda} \right)\,e^{\frac{i\tau \alpha^2}{2\hbar}}\,\exp\left\{\frac{i}{\hbar}\alpha_\nu^2\int_0^\tau\beta^n_\nu\circ\gamma(t)\,dt\right\}\right].
\end{equation}
Let $I_\lambda$ denote the contributing terms to the integral proportional to the $\lambda^{\mathrm{th}}$ order in the expansion of $F$. Following the same process as in the proof of Theorem \ref{thm:main1}, letting $\mu=\alpha\sqrt{\tau}$, we have
\begin{equation}
    I_\lambda = \mu^\lambda\tau^{\frac{\lambda}{2}}\exp\left\{\frac{i}{\hbar}\frac{\mu^2}{2}(1+2\beta_\nu^n)\right\}\exp\left\{\frac{i}{\hbar}\sum_{j=1}^n\mu^j\tau^{\frac{j}{2}+1}\beta^{n-j}\frac{n!}{(n-j)!\,(j+1)!}\right\}.
\end{equation}
Note that for any $n$, the lowest order terms in $\tau$ contributed by the power series of the second exponential are of $\mathcal{O}(\tau^0)$ and $\mathcal{O}(\tau^{3/2})$. Thus, all terms contributed by this series apart from those of $\mathcal{O}(\tau^0)$ vanish in the limit $\lim_{\tau\rightarrow 0}\frac{d}{dt}I_\lambda$, independent of $\lambda$. From here it is clear that the only possible nonzero contribution comes from $I_2$ in the identity component of the exponential power series, i.e. \eqref{eqn:main2initial} reduces to
\begin{equation}
     \frac{d\widetilde{\psi}}{dt} =-\frac{d^2F}{d\beta^2}\int_{-\infty}^{+\infty}d\mu\,\frac{\mu^2}{2}\exp\left\{\frac{i}{\hbar}\mu^2(1+2\beta_\nu^n)\right\}.
\end{equation}
We can integrate and absorb the same prefactors occurring in \eqref{eqn:almostTDSE} into the definition of the pairing to obtain the desired relation, the Schr\"odinger equation for a simple position type deformation.
\end{proof}

\vspace{0.25cm}
\subsection*{Acknowledgements}
The author is thankful to Duncan O'Dell for many useful conversations and his hospitality during much of this work. This work was supported by the EPSRC Centre for Doctoral Training in Topological Design (EP/S02297X/1) and a MITACS Globalink Research Award (IT40628).
\printbibliography
\end{document}